\pdfoutput=1
\documentclass[11pt]{article}
\usepackage{subfig}
\usepackage{amssymb,amsmath,amsthm}
\usepackage{graphics,epsfig}
\usepackage{algorithmic}
\usepackage{algorithm}
\usepackage[margin=1 in]{geometry}

\newtheorem{lem}{Lemma}[section]

\newtheorem{thm}[lem]{Theorem}

\newcounter{constraint}
\title{Scheduling Coflows for Minimizing the Total Weighted Completion Time in Heterogeneous Parallel Networks}
\author{Chi-Yeh~Chen 
\\ Department of Computer Science and Information
Engineering, \\ National Cheng Kung University, \\
Taiwan, ROC. \\
chency@csie.ncku.edu.tw.}

\begin{document}

\maketitle
\begin{abstract}
Coflow is a network abstraction used to represent communication patterns in data centers. The coflow scheduling problem in large data centers is one of the most important $NP$-hard problems. Many previous studies on coflow scheduling mainly focus on the single-core model. However, with the growth of data centers, this single-core model is no longer sufficient. This paper considers the coflow scheduling problem in heterogeneous parallel networks. The heterogeneous parallel network is an architecture based on multiple network cores running in parallel. In this paper, two polynomial-time approximation algorithms are developed for scheduling divisible and indivisible coflows in heterogeneous parallel networks, respectively. Considering the divisible coflow scheduling problem, the proposed algorithm achieve an approximation ratio of $O(\log m/ \log \log m)$ with arbitrary release times, where $m$ is the number of network cores. On the other hand, when coflow is indivisible, the proposed algorithm achieve an approximation ratio of $O\left(m\left(\log m/ \log \log m\right)^2\right)$ with arbitrary release times.

\end{abstract}

\section{Introduction}\label{sec:Introduction}
With the rapid development of cloud computing, large data centers have become the main computing infrastructure.
In large data centers, structured traffic patterns of distributed applications have demonstrated the benefits of application-aware network scheduling~\cite{Chowdhury2014, Chowdhury2015, Zhang2016, Agarwal2018}. 
In addition, the success of data-parallel computing applications such as MapReduce~\cite{Dean2008}, Hadoop~\cite{Shvachko2010, borthakur2007hadoop}, Dryad~\cite{isard2007dryad} and Spark~\cite{zaharia2010spark} has led to a proliferation of related applications~\cite {dogar2014decentralized, chowdhury2011managing}.
In data-parallel computing applications, computations are only processed locally on the machine. However, intermediate data (flows) generated during the computation stage need to be transmitted across different machines during the communication stage for further processing.
Due to the large number of applications, the data center must have sufficient data transmission and scheduling capabilities.
In data transfer for data-parallel computing applications, the interaction of all flows between two groups of machines becomes important.
This collective communication pattern in the data center is abstracted by coflow traffic~\cite{Chowdhury2012}.

Many previous studies on coflow scheduling mainly focus on the single-core model. However, with the growth of data centers, this single-core model is no longer sufficient. In fact, a growing data center will have legacy and new systems coexisting. To improve the efficiency of the network, there will be different generations of network cores running in parallel~\cite{Singh2015, Huang2020}. Therefore, we consider an architecture based on multiple heterogeneous network cores running in parallel (heterogeneous parallel network). The goal of this paper is to schedule coflows in the heterogeneous parallel networks such that the total weighted completion time is minimized. This paper will discuss two problems: the indivisible coflow scheduling problem and the divisible coflow scheduling problem.
In the indivisible coflow scheduling problem, flows in a coflow can only be arranged in the same network core.
However, in the divisible coflow scheduling problem, the flows in a coflow can be arranged in different network cores.

\subsection{Related Work}
The coflow abstraction was first introduced by Chowdhury and Stoica~\cite{Chowdhury2012} to capture communication patterns in data centers.
The coflow scheduling problem generalizes the well-studied concurrent open shop scheduling problem, which has
been shown to be strongly NP-hard~\cite{chen2007supply, garg2007order, leung2007scheduling, mastrolilli2010minimizing, wang2007customer}.
Therefore, we instead look for efficient approximation algorithms rather than exact algorithms.
Since the concurrent open shop problem is NP-hard to approximate within a factor better than $2-\epsilon$ for any $\epsilon>0$~\cite{Sachdeva2013, shafiee2018improved}, the coflow scheduling problem is also NP-hard to approximate within a factor better than $2-\epsilon$~\cite{ahmadi2020scheduling, Bansal2010, Sachdeva2013}.
Since the introduction of the coflow abstraction, many related investigations have been carried out to schedule coflows, e.g. \cite{Chowdhury2014, Chowdhury2015, Qiu2015, zhao2015rapier, shafiee2018improved, ahmadi2020scheduling}.
The first polynomial-time deterministic approximation algorithm was developed by Qiu~\textit{et al.}~\cite{Qiu2015}.
Since then, a series of improved approximation algorithms have been proposed~\cite{Qiu2015, khuller2016brief, shafiee2018improved, ahmadi2020scheduling} and the best approximation ratio achievable in polynomial time has improved from $\frac{64}{3}$ to 4.
Moreover, the best approximation ratio has been improved from $\frac{76}{3}$ to 5, taking into account arbitrary release times.
When each job has multiple coflows and there is a priority order between coflows, Shafiee and Ghaderi~\cite{shafiee2021scheduling} proposed a polynomial-time algorithm with approximation ratio of $O(\mu \log(N)/\log(\log(N)))$, where $\mu$ is the maximum number of coflows in a job and $N$ is the number of servers.
In scheduling a single coflow problem on a heterogeneous parallel network, Huang \textit{et al.}~\cite{Huang2020} proposed an $O(m)$-approximation algorithm, where $m$ is the number of network cores.

\subsection{Our Contributions}
This paper considers the coflow scheduling problem in heterogeneous parallel networks. Our results are as follows:
\begin{itemize}
\item In the indivisible coflow scheduling problem, we propose a $O\left(m\left(\log m/ \log \log m\right)^2\right)$-approximation algorithm with arbitrary release times.

\item In the divisible coflow scheduling problem, we also propose a $O(\log m/ \log \log m)$-approximation algorithm with arbitrary release times. 
\end{itemize}

\subsection{Organization}
The rest of this article is organized as follows. Section \ref{sec:Preliminaries} introduces basic notations and preliminaries. Section \ref{sec:Algorithm2} presents an algorithm for indivisible coflow scheduling. Section \ref{sec:Algorithm1} presents an algorithm for divisible coflow scheduling. Section~\ref{sec:Results} compares the performance of the previous algorithms with that of the proposed algorithm. Section \ref{sec:Conclusion} draws conclusions.
\section{Notation and Preliminaries}\label{sec:Preliminaries}
Given a set of coflows $\mathcal{F}$ and a set of heterogeneous network cores $\mathcal{M}$, the coflow scheduling problem asks for a minimum total weighted completion time, where each coflow has its release time and its positive weight. The heterogeneous parallel networks can be abstracted as a set $\mathcal{M}$ of $m$ giant $N \times N$ non-blocking switchs, with $N$ input links connect to $N$ source servers and $N$ output links connect to $N$ destination servers. Each switch represents a network core. In each network core, all links are assumed to have the same capacity.
Let $s_k$ be the link speed of network core $k$. Each source server or each destination server has $m$ simultaneous links connected to each network core.
Let $\mathcal{I}$ be the source server set and $\mathcal{J}$ be the destination server set. The network core can be seen as a bipartite graph, with $\mathcal{I}$ on one side and $\mathcal{J}$ on the other side. 

A coflow consists of a set of independent flows whose completion time is determined by the last completed flow in the set.
We can use a $N\times N$ demand matrix $D^{(f)}=(d_{ijf})_{i,j=1}^{N}$ to represent the coflow $f\in \mathcal{F}$ where $d_{ijf}$ denote the size of the flow to be transferred from input $i$ to output $j$ in coflow $f$.
We also can use a triple $(i, j, f)$ to represent a flow, where $i \in \mathcal{I}$, $j \in \mathcal{J}$ and $f \in \mathcal{F}$.
For simplicity, we assume that all flows in a coflow arrive at the system at the same time (as shown in~\cite{Qiu2015}).
Let $w_{f}$, $r_f$ and $C_f$ denote the weight of coflow $f$, the release time of coflow $f$ and the completion time of coflow $f$, respectively.
The goal is to minimize the total weighted completion time of the coflow $\sum_{f\in \mathcal{F}} w_{f}C_{f}$.
We consider two problems: the indivisible coflow scheduling problem and the divisible coflow scheduling problem.
In the indivisible coflow scheduling problem, flows in a coflow can only be arranged in the same network core.
However, in the divisible coflow scheduling problem, the flows in a coflow can be arranged in different network cores.

\section{Approximation Algorithm for Indivisible Coflow Scheduling}\label{sec:Algorithm2}
In this section, as in~\cite{li2020scheduling}, we first give $O\left(m\left(\log m/ \log \log m\right)^2\right)$-approximation to minimize the makepan scheduling problem on heterogeneous network cores. We then convert the goal of minimizing the total weighted completion time to minimizing makespan at the loss of a constant factor.
For every coflow $f$ and input port $i$, let $L_{if}=\sum_{j=1}^{N}d_{ijf}$ be the total amount of data that coflow $f$ needs to transmit through the input port $i$. Moreover, let $L_{jf}=\sum_{i=1}^{N}d_{ijf}$ be the total amount of data that coflow $f$ needs to transmit through the output port $j$.
We can formulate our problem as the following linear programming relaxation.

\begin{subequations}\label{incoflow:main}
\begin{align}
&    \text{min} && E                           && \tag{\ref{incoflow:main}}\\ 
&    \text{s.t.}. && \sum_{k\in \mathcal{M}} x_{kf}=1 && \forall f\in \mathcal{F} \label{incoflow:a}\\ 
&         && \sum_{k\in \mathcal{M}} \frac{x_{kf}L_{if}}{s_k}\leq C_{f}  && \forall f\in \mathcal{F}, \forall i\in \mathcal{I} \label{incoflow:b}\\
&         && \sum_{k\in \mathcal{M}} \frac{x_{kf}L_{jf}}{s_k}\leq C_{f}  && \forall f\in \mathcal{F}, \forall j\in \mathcal{J} \label{incoflow:c}\\
&         && \frac{1}{s_k}\sum_{f\in \mathcal{F}} x_{kf}L_{if}  \leq E && \forall i\in \mathcal{I}, \forall k\in \mathcal{M} \label{incoflow:f}\\ 
&         && \frac{1}{s_k}\sum_{f\in \mathcal{F}} x_{kf}L_{jf}  \leq E && \forall j\in \mathcal{J}, \forall k\in \mathcal{M} \label{incoflow:g}\\ 
&         && C_{f}\leq E && \forall f\in \mathcal{F} \label{incoflow:h}\\ 
&         && x_{kf}, C_{f}\geq 0 && \forall f\in \mathcal{F}, \forall k\in \mathcal{M} \label{incoflow:i}
\end{align}   
\end{subequations}

In the LP (\ref{incoflow:main}), $x_{kf}$ indicates whether coflow $f$ is scheduled on network core $k$, $E$ is the makespan of the schedule and $C_{f}$ is the completion time of coflow $f$ in the schedule.
The constraint~(\ref{incoflow:a}) requires scheduling for each coflow $f$.
The constraint~(\ref{incoflow:b}) (similarly the constraint~(\ref{incoflow:c})) is that the completion time of $f$ occurring on input port $i$ (output port $j$) is at least the transfer time on the network cores allocated to it. 
The constraint~(\ref{incoflow:f}) (similarly the constraint~(\ref{incoflow:g})) says that, for each network core $k$ and input port $i$ (output port $j$), the makespan $E$ is at least the total transfer time that occurs on input port $i$ (output port $j$) in all coflows assigned to $k$.
The constraint~(\ref{incoflow:h}) states that the makespan $E$ is at least the completion time of any coflow $f$.
The constraint~(\ref{incoflow:i}) requires the $x$ and $C$ variables to be non-negative.

The following rounding method follows the method proposed by Li~\cite{li2020scheduling}.
The optimal solution to LP (\ref{incoflow:main}) is the lower bound on the makespan of any valid schedule.
We assume that $m$ is large enough. 
Given an instance of the makepan scheduling problem, we will first preprocess the instance like~\cite{li2020scheduling} to contain only a small number of groups.
The first stage of the pre-processing step discards all network cores that are at most $1/m$ times the speed of the fastest network core.
Since there are $m$ network cores, the total speed of discarded network cores is at most that of the fastest network core.
That is, for the same amount of transferred data, the transfer time of the fastest network core is at most the transfer time of using all discarded network cores; however, this will increase the makepan by a factor of 2.
So we can move the $x$-values of all discarded network cores to the $x$-values of the fastest network cores.
Therefore, we can assume that all network cores are faster than $1/m$ times the speed of the fastest network core. We also can normalize the speed of all network cores to $[1,m)$.

The second stage of the pre-processing step divides the network cores into groups, each group containing network cores of similar speed.
Let $\gamma= \log m/\log \log m$.
The network cores are divided into $K$ groups $M_1, M_2, \ldots, M_K$, where $M_k$ contains network cores with speed in $[\gamma^{k-1},\gamma^{k})$ and $K= \left\lceil \log_{\gamma} m \right\rceil= O(\log m/\log \log m)$.
For a subset $M_k\subseteq M$ of network cores, let 
\begin{eqnarray*}
s(M_k)=\sum_{u\in M_k} s_u
\end{eqnarray*}
be the total speed of network cores in $M_k$. For $M_k\subseteq M$ and $f\in \mathcal{F}$, let
\begin{eqnarray*}
x_{M_kf}=\sum_{u\in M_k} x_{uf}
\end{eqnarray*}
be the total fraction of coflow $f$ assigned to network cores in $M_k$. For any coflow $f$, let $\ell_f$ be the largest integer $\ell$ such that $\sum_{k=\ell}^{K} x_{M_kf}\geq 1/2$.
That is, the largest group index $\ell$ such that at least 1/2 of $f$ is allocated to network cores in groups $M_\ell, \ldots, M_K$.
Then, let $r(f)$ be the index $k\in [\ell_f, K]$ that maximizes $s(M_k)$. That is, $r(f)$ is the index of the group with the highest total speed among the groups $\ell_f$ to $K$. Using the $r(f)$ value, we can run the list algorithm (\textbf{Algorithm}~\ref{Alg1}) to get the index of the allocated network core. The algorithm is to find the least loaded network core and assign coflow to it. The proposed algorithm has the following lemmas.

\begin{algorithm}
\caption{coflow-makespan-list-scheduling}
    \begin{algorithmic}[1]
		    \REQUIRE two vectors $\bar{C}\in \mathbb{R}_{\scriptscriptstyle \geq 0}^{n}$ and $\bar{x}\in \mathbb{R}_{\scriptscriptstyle \geq 0}^{n\times m}$
				\STATE let $load_{I}(i,h)$ be the load on the $i$-th input port of the network core $h$
				\STATE let $load_{O}(j,h)$ be the load on the $j$-th output port of the network core $h$
				\STATE let $\mathcal{A}_h$ be the set of coflows allocated to network core $h$
				\STATE both $load_{I}$ and $load_{O}$ are initialized to zero and $\mathcal{A}_h=\emptyset$ for all $h\in [1, m]$
				\FOR{every coflow $f\in \mathcal{F}$ in non-decreasing order of $\bar{C}_f$, breaking ties arbitrarily}
				    \STATE $x_{M_kf}=\sum_{u\in M_k} x_{uf}$ for all $k=1,\ldots, K$
				    \STATE $\ell_f=\max_{\ell\in [1,K]} \ell$ s.t. $\sum_{k=\ell}^{K} x_{M_kf}\geq 1/2$
						\STATE $r(f)=\arg_{M_k:\ell_f\leq k\leq K} \max s(M_k)$
						\STATE $h^*=\arg \min_{h\in M_{r(f)}}\frac{1}{s_h}\left(\max_{i,j\in [1,N]}load_{I}(i,h)+\right.$ $\left.load_{O}(j,h)+L_{if}+L_{jf}\right)$
						\STATE $\mathcal{A}_{h^*}=\mathcal{A}_{h^*}\cup \left\{f\right\}$
						\STATE $load_{I}(i,h^*)=load_{I}(i,h^*)+L_{if}$ and $load_{O}(j,h^*)=load_{O}(j,h^*)+L_{jf}$ for all $i,j\in [1,N]$
				\ENDFOR
   \end{algorithmic}
\label{Alg1}
\end{algorithm}

\begin{lem}\label{lem:lem2}
For any port $i\in \mathcal{I}$, we have $\sum_{f\in \mathcal{F}}\frac{L_{if}}{s(M_{r(f)})}\leq 2KE$.
\end{lem}
\begin{proof}
Since $\sum_{k=\ell_f}^{K}x_{M_k,f}\geq \frac{1}{2}$ for any coflow $f$, we have 
\begin{eqnarray*}
\sum_{k=1}^{K}\frac{x_{M_kf}}{s(M_k)}\geq \sum_{k=\ell_f}^{K}\frac{x_{M_kf}}{s(M_k)}\geq \frac{1}{2s(M_k)}.
\end{eqnarray*}
According to the above inequality, we have
\begin{eqnarray*}
\sum_{f\in \mathcal{F}} \frac{L_{if}}{s(M_{r(f)})} & \leq & 2\sum_{f\in \mathcal{F}}L_{if}\sum_{k=1}^{K}\frac{x_{M_kf}}{s(M_{k})} \\
   & = & 2\sum_{k=1}^{K}\frac{1}{s(M_{k})}\sum_{f\in \mathcal{F}}L_{if}x_{M_kf} \\
	 & \leq & 2\sum_{k=1}^{K} E \\
	 & = & 2KE.
\end{eqnarray*}
The last inequality is due to constraint~(\ref{incoflow:f}). Since $\sum_{f\in F} x_{kf}L_{if} \leq s_k E$ for every $i\in M_k$, we have $\sum_{f\in \mathcal{F}}L_{if}x_{M_k,f}\leq s(M_{k})E$
\end{proof}

\begin{lem}\label{lem:lem3}
For any port $j\in \mathcal{J}$, we have $\sum_{f\in \mathcal{F}}\frac{L_{jf}}{s(M_{r(f)})}\leq 2KE$.
\end{lem}
\begin{proof}
The proof is similar to that of lemma~\ref{lem:lem2}.
\end{proof}

\begin{lem}\label{lem:lem4}
Let $\bar{E}$ be an optimal solution to the linear program (\ref{incoflow:main}), and let $\tilde{E}$ denote the makespan in the schedule found by coflow-makespan-list-scheduling. We have
\begin{eqnarray*}
\tilde{E}\leq 4m\gamma K\bar{E}.
\end{eqnarray*}
\end{lem}
\begin{proof}
Assume the last completed flow in a coflow $f$ is sent via link $(i, j)$ of network core $k$. We have
\begin{eqnarray}\label{eq:1}
\tilde{E} & \leq & \sum_{f\in \mathcal{F}} \frac{L_{if}+L_{jf}}{s_k} \label{eq:1-1}\\
          & =    & \sum_{f\in \mathcal{F}} \frac{s(M_{r(f)})}{s_k}\frac{L_{if}+L_{jf}}{s(M_{r(f)})} \label{eq:1-2}\\
          & \leq & m\gamma \sum_{f\in \mathcal{F}} \frac{L_{if}+L_{jf}}{s(M_{r(f)})} \label{eq:1-3}\\
          & \leq & 4m\gamma K\bar{E} \label{eq:1-4}
\end{eqnarray}
The inequality~(\ref{eq:1-1}) is due to the worst case where all coflows are assigned to the same core.
The inequality~(\ref{eq:1-3}) is dut to $\frac{s(M_{r(f)})}{s_k} \leq m \gamma$.
The inequality~(\ref{eq:1-4}) is based on lemma~\ref{lem:lem2} and lemma~\ref{lem:lem3}.
\end{proof}

According to lemma~\ref{lem:lem4}, we have the following theorem:
\begin{thm}\label{thm:thm1}
When the speed of network core is between one and $m$, the indivisible coflow schedule has makespan at most $4m\gamma K\bar{E}=O\left(m\left(\log m/ \log \log m\right)^2\right)\bar{E}$.
\end{thm}

Due to the discarded network cores, we have the following theorem:
\begin{thm}\label{thm:thm2}
When the speed of network core is arbitrary, the indivisible coflow schedule has makespan at most $8m\gamma K\bar{E}=O\left(m\left(\log m/ \log \log m\right)^2\right)\bar{E}$.
\end{thm}

\subsection{An extension for Total Weighted Completion Time}\label{sec:Algorithm2-1}
This section gives an $O\left(m\left(\log m/ \log \log m\right)^2\right)$ -approximation algorithm for minimizing total weighted completion time, which is based on combining our algorithm for minimizing makespan. 
Without loss of generality, we assume that $L_{if}/s_k\geq 1$ and $L_{jf}/s_k\geq 1$ for all $f\in \mathcal{F}$, $k\in \mathcal{M}$, $i\in \mathcal{I}$ and $j\in \mathcal{J}$. Let $s_{min}= \min_{k\in \mathcal{M}} s_k$. We have 
\begin{eqnarray*}
L=\log \left(\max_{f\in \mathcal{F}} r_f+\frac{1}{s_{min}}\max\left\{\max_{i\in \mathcal{I}}\sum_{f\in \mathcal{F}} L_{if}, \max_{j\in \mathcal{J}}\sum_{f\in \mathcal{F}} L_{jf}\right\}\right).
\end{eqnarray*}
First, we divide the time horizon into increasing time intervals: $[1,2], (2, 4], (4, 8], \ldots, (2^{L-1}, 2^{L}]$. Let $\tau_{l}=2^l$ where $l=0, 1, \ldots, L$.
We can formulate our problem as the following linear programming relaxation.

\begin{subequations}\label{incoflow:main:wc}
\begin{align}
&    \text{min} && \sum_{f\in \mathcal{F}}w_f C_{f}                           && \tag{\ref{incoflow:main:wc}}\\ 
&    \text{s.t.}. && \sum_{k\in \mathcal{M}} \sum_{l=1}^{L} x_{kfl}=1, && \forall f\in \mathcal{F}  \label{incoflow:wc:a}\\ 
&         && \sum_{k\in \mathcal{M}} \frac{L_{if}}{s_k}\sum_{l=1}^{L}x_{kfl}\leq C_{f}-r_f  && \forall f\in \mathcal{F}, \forall i\in \mathcal{I} \label{incoflow:wc:b}\\
&         && \sum_{k\in \mathcal{M}} \frac{L_{jf}}{s_k}\sum_{l=1}^{L}x_{kfl}\leq C_{f}-r_f  && \forall f\in \mathcal{F}, \forall j\in \mathcal{J} \label{incoflow:wc:c}\\
&         && \frac{1}{s_k}\sum_{u=1}^{l} \sum_{f\in \mathcal{F}} x_{kfu}L_{if}  \leq \tau_{l} && \forall i\in \mathcal{I}, \forall k\in \mathcal{M}, \forall l\in [1, L] \label{incoflow:wc:f}\\ 
&         && \frac{1}{s_k}\sum_{u=1}^{l} \sum_{f\in \mathcal{F}} x_{kfu}L_{jf}  \leq \tau_{l} && \forall j\in \mathcal{J}, \forall k\in \mathcal{M}, \forall l\in [1, L] \label{incoflow:wc:g}\\ 
&         && \sum_{l=1}^{L}\tau_{l-1} \sum_{k\in \mathcal{M}} x_{kfl} \leq C_{f}, && \forall f\in \mathcal{F} \label{incoflow:wc:h}\\ 
&         && x_{kfl}, C_{f}\geq 0 && \forall f\in \mathcal{F}, \forall k\in \mathcal{M} \label{incoflow:wc:i}
\end{align}   
\end{subequations}

In the LP (\ref{incoflow:main:wc}), $x_{kfl}$ indicates whether or not coflow $f$ completes on network core $k$ in the $l$-th interval (from $2^{l-1}$ to $2^{l}$) and $C_{f}$ is the completion time of coflow $f$ in the schedule.
The constraint~(\ref{incoflow:wc:a}) requires all the coflows must be assigned to run at some network core.
The constraint~(\ref{incoflow:wc:b}) (similarly the constraint~(\ref{incoflow:wc:c})) is that the time required to transmit coflow $f$ on input port $i$ (output port $j$) cannot exceed the time period between its release and completion time. 
The constraints~(\ref{incoflow:wc:f}) and (\ref{incoflow:wc:g}) represent capacity limits to time $\tau_{l}$.
Since $\tau_{l-1}$ is the lower bound on the completion time of coflows completed within the interval $l$, the constraint~(\ref{incoflow:wc:h}) is a lower bound to the completion time of coflow.
The constraint~(\ref{incoflow:wc:i}) requires the $x$ and $C$ variables to be non-negative.

Our algorithm coflow-driven-list-scheduling (described in Algorithm~\ref{Alg2}) is as follows. 
Given a set of coflow $\mathcal{F}$, an optimal solution $\bar{C}$ and $\bar{x}$ can be obtained by the linear program (\ref{incoflow:main:wc}). 
Lines 1-5 schedule all coflow into time intervals and normalize the value of $\bar{x}$. 
Lines 6-19 are the coflow-makespan-list-scheduling algorithm, which schedules the coflows in the corresponding time interval to each network core.
Lines 20-36 transmit all coflow, which is modified from Shafiee and Ghaderi's algorithm ~\cite{shafiee2018improved}.

\begin{algorithm}
\caption{coflow-driven-list-scheduling}
    \begin{algorithmic}[1]
		    \REQUIRE two vectors $\bar{C}\in \mathbb{R}_{\scriptscriptstyle \geq 0}^{n}$ and $\bar{x}\in \mathbb{R}_{\scriptscriptstyle \geq 0}^{n\times m}$
				\FOR{every coflow $f\in \mathcal{F}$}
				    \STATE $q(f)=\min_{l\in [1,L]} l$ s.t. $\sum_{l=1}^{q(f)}\sum_{k\in \mathcal{M}} x_{kfl}\geq 1/2$ and $\bar{C}_f\leq 2^{q(f)}$
						\STATE $\mathcal{F}_{q(f)}=\mathcal{F}_{q(f)} \cup \left\{f\right\}$
						\STATE $\alpha_f =\sum_{l=1}^{q(f)}\sum_{k\in \mathcal{M}} \bar{x}_{kfl}$ and $\tilde{x}_{kf}=\sum_{l=1}^{q(f)} \frac{\bar{x}_{kfl}}{\alpha_f}$ for all $k\in \mathcal{M}$
				\ENDFOR
				\FOR{$l\in [1,L]$}
				    \STATE let $load_{I}(i,h)$ be the load on the $i$-th input port of the network core $h$
				    \STATE let $load_{O}(j,h)$ be the load on the $j$-th output port of the network core $h$
				    \STATE let $\mathcal{A}_{lh}$ be the set of coflows allocated to network core $h$ in the $l$-th interval
				    \STATE both $load_{I}$ and $load_{O}$ are initialized to zero and $\mathcal{A}_{lh}=\emptyset$ for all $h\in [1, m]$
				    \FOR{every coflow $f\in \mathcal{F}_{l}$ in non-decreasing order of $\bar{C}_f$, breaking ties arbitrarily}
						    \STATE $\tilde{x}_{M_kf}=\sum_{u\in M_k} \tilde{x}_{uf}$ for all $k=1,\ldots, K$
				        \STATE $\ell_f=\max_{\ell\in [1,K]} \ell$ s.t. $\sum_{k=\ell}^{K} \tilde{x}_{M_kf}\geq 1/2$
						    \STATE $r(f)=\arg_{M_k:\ell_f\leq k\leq K} \max s(M_k)$
						    \STATE $h^*=\arg \min_{h\in M_{r(f)}}\frac{1}{s_h}\left(\max_{i,j\in [1,N]}load_{I}(i,h)+\right.$ $\left.load_{O}(j,h)+L_{if}+L_{jf}\right)$
						    \STATE $\mathcal{A}_{lh^*}=\mathcal{A}_{lh^*}\cup \left\{f\right\}$
						    \STATE $load_{I}(i,h^*)=load_{I}(i,h^*)+L_{if}$ and $load_{O}(j,h^*)=load_{O}(j,h^*)+L_{jf}$ for all $i,j\in [1,N]$
				    \ENDFOR
				\ENDFOR
		    \FOR{each $k\in \mathcal{M}$ do in parallel}
    				\FOR{$l\in [1,L]$}
				        \STATE wait until the first coflow is released
						    \WHILE{there is some incomplete flow}
						        \STATE for all $f\in \mathcal{A}_{lk}$, list the released and incomplete flows respecting the non-decreasing order in $\bar{C}_f$
								    \STATE let $L$ be the set of flows in the list
                    \FOR{every flow $(i, j, f)\in L$}
								    		\IF{the link $(i, j)$ is idle}
								    		    \STATE schedule flow $f$
								    		\ENDIF
								    \ENDFOR
								    \WHILE{no new flow is completed or released}
								        \STATE transmit the flows that get scheduled in line 28 at maximum rate $s_k$.
								    \ENDWHILE
						    \ENDWHILE
				    \ENDFOR						
				\ENDFOR	
   \end{algorithmic}
\label{Alg2}
\end{algorithm}

The following schedule method follows the method proposed by Chudak and Shmoys~\cite{CHUDAK1999323}.
For any coflow $f\in \mathcal{F}$, let $q(f)$ be the the minimum value of $q$ such that both $\sum_{l=1}^{q(f)}\sum_{k\in \mathcal{M}} x_{kfl}\geq 1/2$ and $\bar{C}_f\leq 2^{q(f)}$ are satisfied. 
We set $\mathcal{F}_{l}=\left\{f|q(f)=l\right\}$ and construct a schedule for each subset $\mathcal{F}_{l}$ respectively.
Let $\alpha_f$ be the total fraction of coflow $f$ over all network cores in the first $q(f)$ intervals with respect to solution $\bar{x}$:
\begin{eqnarray}\label{eq:3}
\alpha_f =\sum_{l=1}^{q(f)}\sum_{k\in \mathcal{M}} \bar{x}_{kfl}.
\end{eqnarray}
We set a feasible solution $\tilde{x}$ from the optimal solution $\bar{x}$:
\begin{eqnarray}\label{eq:4}
\tilde{x}_{kf}=\sum_{l=1}^{q(f)} \frac{\bar{x}_{kfl}}{\alpha_f}
\end{eqnarray}
for all $f\in \mathcal{F}$ and $k\in \mathcal{M}$.

Fix some $l=1,\ldots, L$, and consider the coflows in $\mathcal{F}_{l}$. 
We can construct a scheduling fragment for $\mathcal{F}_{l}$ of length $\bar{R}2^{l+1}$, where $\bar{R}$ is the performance guarantee of the proposed approximation algorithm for the makespan objective. This fragment shall be run from time $\bar{R}(1+2+\cdots+2^l)$ to $\bar{R}(1+2+\cdots+2^{l+1})$. Therefore, each coflow $f\in \mathcal{F}_{l}$ completes at most $4\bar{R}2^{l}$ before. Since $l$ is the minimum value for $\bar{C}_f\leq 2^{l}$, we have $2^{l}\leq 2\bar{C}_f$. Since $\sum_{u=1}^{l}\sum_{k\in \mathcal{M}} x_{kfu}\geq 1/2$, we have
\begin{eqnarray*}\label{eq:5}
\tau_{l-1}/2 & \leq &\tau_{l-1}\left(\sum_{u=l}^{L}\sum_{k\in \mathcal{M}} x_{kfu}\right) \\
& \leq & \sum_{u=l}^{L}\sum_{k\in \mathcal{M}} \tau_{u-1} x_{kfu} \\
& \leq & \sum_{u=1}^{L}\sum_{k\in \mathcal{M}} \tau_{u-1} x_{kfu} \\
& \leq & \bar{C}_f.
\end{eqnarray*}
We have $2^{l}\leq 4\bar{C}_f$ and proved the following theorem.
\begin{thm}
The coflow-driven-list-scheduling has an approximation ratio of, at most, $64m\gamma K=O\left(m\left(\log m/ \log \log m\right)^2\right)$.
\end{thm}

\section{Approximation Algorithm for Divisible Coflow Scheduling}\label{sec:Algorithm1}
This section considers the divisible coflow scheduling problem. 
The method is similar to schedule indivisible coflow, the difference is that it is scheduled at the flow level. First, we consider the minimizing makespan problem. We can formulate our problem as the following linear programming relaxation.

\begin{subequations}\label{coflow:main}
\begin{align}
&    \text{min} && E                           && \tag{\ref{coflow:main}}\\ 
&    \text{s.t.} && \sum_{k\in \mathcal{M}} x_{kijf}=1, && \forall f\in \mathcal{F}, \forall i\in \mathcal{I}, \forall j\in \mathcal{J} \label{coflow:a}\\ 
&         && \sum_{k\in \mathcal{M}} \frac{d_{ijf} x_{kijf}}{s_k}\leq C_{ijf}  && \forall f\in \mathcal{F}, \forall i\in \mathcal{I}, \forall j\in \mathcal{J} \label{coflow:b}\\
&         && \frac{1}{s_k}\sum_{f\in \mathcal{F}} \sum_{j\in \mathcal{J}} d_{ijf} x_{kijf} \leq E && \forall i\in \mathcal{I}, \forall k\in \mathcal{M} \label{coflow:f}\\ 
&         && \frac{1}{s_k}\sum_{f\in \mathcal{F}} \sum_{i\in \mathcal{I}} d_{ijf} x_{kijf} \leq E && \forall j\in \mathcal{J}, \forall k\in \mathcal{M} \label{coflow:g}\\ 
&         && C_{ijf}\leq E && \forall f\in \mathcal{F}, \forall i\in \mathcal{I}, \forall j\in \mathcal{J} \label{coflow:h}\\ 
&         && x_{kijf}, C_{ijf}\geq 0 && \forall f\in \mathcal{F}, \forall i\in \mathcal{I}, \forall j\in \mathcal{J}, \forall k\in \mathcal{M} \label{coflow:i}
\end{align}   
\end{subequations}

In the LP (\ref{coflow:main}), $x_{kijf}$ indicates whether flow $(i, j, f)$ is scheduled on network core $k$, $E$ is the makespan of the schedule and $C_{ijf}$ is the completion time of $(i, j, f)$ in the schedule.
The constraint~(\ref{coflow:a}) requires scheduling for each flow $(i, j, f)$.
The constraint~(\ref{coflow:b}) is that the completion time of $(i, j, f)$ occurring on link $(i, j)$ is at least the transfer time on the network cores allocated to it. 
The constraint~(\ref{coflow:f}) (similarly the constraint~(\ref{coflow:g})) says that, for each network core $k$ and input port $i$ (output port $j$), the makespan $E$ is at least the total transfer time that occurs on input port $i$ (output port $j$) in all flows assigned to $k$.
The constraint~(\ref{coflow:h}) states that the makespan $E$ is at least the completion time of any flow $(i, j f)$.
The constraint~(\ref{coflow:i}) requires the $x$ and $C$ variables to be non-negative.

Following the steps in Section~\ref{sec:Algorithm2}, we divide the network cores into groups. 
We have $\gamma= \log m/\log \log m$ and $K$ groups $M_1, M_2, \ldots, M_K$, where $M_k$ contains network cores with speed in $[\gamma^{k-1},\gamma^{k})$.
For a subset $M_k\subseteq M$ of network cores, we also have
\begin{eqnarray*}
s(M_k)=\sum_{u\in M_k} s_u.
\end{eqnarray*}
For $M_k\subseteq M$, let
\begin{eqnarray*}
x_{M_kijf}=\sum_{u\in M_k} x_{uijf}
\end{eqnarray*}
be the total fraction of flow $(i, j, f)$ assigned to network cores in $M_k$. For any flow $(i, j, f)$, let $\ell_{ijf}$ be the largest integer $\ell$ such that $\sum_{k=\ell}^{K} x_{M_kijf}\geq 1/2$.
Then, let $r(i, j, f)$ be the index $k\in [\ell_{ijf}, K]$ that maximizes $s(M_k)$. Using the $r(i, j, f)$ value, we can run the list algorithm (\textbf{Algorithm}~\ref{Alg3}) to get the index of the allocated network core. 
The algorithm is to find the least loaded network core and assign flow to it. The proposed algorithm has the following lemmas.

\begin{algorithm}
\caption{flow-makespan-list-scheduling}
    \begin{algorithmic}[1]
		    \REQUIRE two vectors $\bar{C}\in \mathbb{R}_{\scriptscriptstyle \geq 0}^{n}$ and $\bar{x}\in \mathbb{R}_{\scriptscriptstyle \geq 0}^{n\times m}$
				\STATE let $load_{I}(i,h)$ be the load on the $i$-th input port of the network core $h$
				\STATE let $load_{O}(j,h)$ be the load on the $j$-th output port of the network core $h$
				\STATE let $\mathcal{A}_h$ be the set of coflows allocated to network core $h$
				\STATE both $load_{I}$ and $load_{O}$ are initialized to zero and $\mathcal{A}_h=\emptyset$ for all $h\in [1, m]$
				\FOR{every flow $(i, j, f)$ in non-decreasing order of $\bar{C}_{ijf}$, breaking ties arbitrarily}
				    \STATE $x_{M_kijf}=\sum_{u\in M_k} x_{uijf}$ for all $k=1,\ldots, K$
				    \STATE $\ell_{ijf}=\max_{\ell\in [1,K]} \ell$ s.t. $\sum_{k=\ell}^{K} x_{M_kijf}\geq 1/2$
						\STATE $r(i, j, f)=\arg_{M_k:\ell_{ijf}\leq k\leq K} \max s(M_k)$
						\STATE $h^*=\arg \min_{h\in M_{r(i, j, f)}}\frac{1}{s_h}\left(load_{I}(i,h)+load_{O}(j,h)\right)$
						\STATE $\mathcal{A}_{h^*}=\mathcal{A}_{h^*}\cup \left\{(i, j, f)\right\}$
						\STATE $load_{I}(i,h^*)=load_{I}(i,h^*)+d_{ijf}$ and $load_{O}(j,h^*)=load_{O}(j,h^*)+d_{ijf}$
				\ENDFOR
   \end{algorithmic}
\label{Alg3}
\end{algorithm}

\begin{lem}\label{lem:lem21}
For every flow $(i, j, f)$, and any network core $k\in M_{\ell_f}$, we have $\frac{d_{ijf}}{s_k}\leq 2 \gamma \sum_{k'\in M}\frac{d_{ijf}x_{k'ijf}}{s_k'}$.
\end{lem}
\begin{proof}
Since $\sum_{k=\ell_{ijf}+1}^{K} x_{M_kijf}< 1/2$, we have $\sum_{k=1}^{\ell_{ijf}} x_{M_kijf}> 1/2$. Moreover, since $\sum_{k'\in \cup_{k=1}^{\ell_{ijf}} M_k} x_{k'ijf}\geq 1/2$ and $\frac{1}{s_k'}\geq \gamma^{-\ell_{ijf}}$ for every $k'\in \cup_{k=1}^{\ell_{ijf}} M_k$, we also have $\sum_{k'\in M}\frac{x_{k'ijf}}{s_k'}\geq \sum_{k'\in \cup_{k=1}^{\ell_{ijf}} M_k}\frac{x_{k'ijf}}{s_k'}\geq \gamma^{-\ell_{ijf}}/2$.

Since $k$ is in group $r(i, j, f)\geq \ell_{ijf}$, $k$ has speed at least $\gamma^{\ell_{ijf}-1}$ and thus $\frac{1}{s_k}\geq \gamma^{1-\ell_{ijf}}$.
Therefore, we have 
\begin{eqnarray*}
\frac{d_{ijf}}{s_k} & \leq & d_{ijf} \gamma^{1-\ell_{ijf}} \\
                    & \leq & 2 \gamma \sum_{k'\in M}\frac{d_{ijf}x_{k'ijf}}{s_k'}.
\end{eqnarray*}
\end{proof}

\begin{lem}\label{lem:lem22}
For any port $i\in \mathcal{I}$, we have $\sum_{f\in \mathcal{F}}\frac{L_{if}}{s(M_{r(i, j, f)})}\leq 2KE$.
\end{lem}
\begin{proof}
Since $\sum_{k=\ell_f}^{K}x_{M_kijf}\geq \frac{1}{2}$ for any flow $(i, j, f)$, we have 
\begin{eqnarray*}
\sum_{k=1}^{K}\frac{x_{M_kijf}}{s(M_k)}\geq \sum_{k=\ell_{ijf}}^{K}\frac{x_{M_kijf}}{s(M_k)}\geq \frac{1}{2s(M_k)}.
\end{eqnarray*}
According to the above inequality, we have
\begin{eqnarray*}
\sum_{f\in \mathcal{F}}\frac{L_{if}}{s(M_{r(i, j, f)})} & = & \sum_{f\in \mathcal{F}}\sum_{j\in \mathcal{J}} \frac{d_{ijf}}{s(M_{r(i, j, f)})} \\
   & \leq & 2\sum_{f\in \mathcal{F}}\sum_{j\in \mathcal{J}}d_{ijf}\sum_{k=1}^{K}\frac{x_{M_kijf}}{s(M_{k})} \\
	 & =    & 2\sum_{k=1}^{K}\frac{1}{s(M_{k})}\sum_{f\in \mathcal{F}}\sum_{j\in \mathcal{J}} d_{ijf}x_{M_kijf} \\
	 & \leq & 2\sum_{k=1}^{K} E \\
	 & =    & 2KE.
\end{eqnarray*}
The last inequality is due to constraint~(\ref{coflow:f}). Since 
$\sum_{f\in \mathcal{F}} \sum_{j\in \mathcal{J}} d_{ijf} x_{kijf} \leq s_k E$ for every $i\in M_k$, we have $\sum_{f\in \mathcal{F}} \sum_{j\in \mathcal{J}} d_{ijf} x_{M_kijf} \leq s(M_{k}) E$.
\end{proof}

\begin{lem}\label{lem:lem23}
For any port $j\in \mathcal{J}$, we have $\sum_{f\in \mathcal{F}}\frac{L_{jf}}{s(M_{r(i, j, f)})}\leq 2KE$.
\end{lem}
\begin{proof}
The proof is similar to that of lemma~\ref{lem:lem22}.
\end{proof}

\begin{lem}\label{lem:lem24}
Let $\bar{E}$ be an optimal solution to the linear program (\ref{coflow:main}), and let $\tilde{E}$ denote the makespan in the schedule found by flow-makespan-list-scheduling. We have
\begin{eqnarray*}
\tilde{E}\leq (4K+2\gamma)\bar{E}.
\end{eqnarray*}
\end{lem}
\begin{proof}
Assume the last completed flow $(i, j, f)$ is sent via the network core $k$. We have
\begin{eqnarray}\label{eq:21}
\tilde{E} & \leq & \sum_{f\in \mathcal{F}} \frac{L_{if}+L_{jf}}{s(M_{r(f)})}+\frac{d_{ijf}}{s_k} \label{eq:21-1}\\
          & \leq & 4K\bar{E}+\frac{d_{ijf}}{s_k} \label{eq:21-2}\\
					& \leq & 4K\bar{E}+2 \gamma \sum_{k'\in M}\frac{d_{ijf}x_{k'ijf}}{s_k'} \label{eq:21-3}\\
					& \leq & 4K\bar{E}+2 \gamma \bar{E} \label{eq:21-4}\\
					& =    & (4K+2\gamma)\bar{E}. \label{eq:21-5}
\end{eqnarray}
The first term of inequality~(\ref{eq:21-1}) is due to all links $(i, j)$ in the network cores are busy from zero to the start of the last completed flow $(i, j, f)$.
The inequality~(\ref{eq:21-2}) is based on lemma~\ref{lem:lem22} and lemma~\ref{lem:lem23}.
The inequality~(\ref{eq:21-3}) is based on lemma~\ref{lem:lem21}.
The inequality (\ref{eq:21-4}) is obtained by constraints~(\ref{coflow:b}) and (\ref{coflow:h}) in the linear program (\ref{coflow:main}).
\end{proof}

According to lemma~\ref{lem:lem24}, we have the following theorem:
\begin{thm}\label{thm:thm21}
When the speed of network core is between one and $m$, the divisible coflow schedule has makespan at most $(4K+2\gamma)\bar{E}=O(\log m/ \log \log m)\bar{E}$.
\end{thm}

Due to the discarded network cores, we have the following theorem:
\begin{thm}\label{thm:thm22}
When the speed of network core is arbitrary, the divisible coflow schedule has makespan at most $(8K+4\gamma)\bar{E}=O(\log m/ \log \log m)\bar{E}$.
\end{thm}

\begin{algorithm}
\caption{flow-driven-list-scheduling}
    \begin{algorithmic}[1]
		    \REQUIRE two vectors $\bar{C}\in \mathbb{R}_{\scriptscriptstyle \geq 0}^{n}$ and $\bar{x}\in \mathbb{R}_{\scriptscriptstyle \geq 0}^{n\times m}$
				\FOR{every coflow $f\in \mathcal{F}$}
				    \STATE $q(i, j, f)=\min_{l\in [1,L]} l$ s.t. $\sum_{l=1}^{q(f)}\sum_{k\in \mathcal{M}} x_{kijfl}\geq 1/2$ and $\bar{C}_{ijf}\leq 2^{q(f)}$
						\STATE $\mathcal{F}_{q(i, j, f)}=\mathcal{F}_{q(i, j, f)} \cup \left\{(i, j, f)\right\}$
						\STATE $\alpha_{ijf} =\sum_{l=1}^{q(i, j, f)}\sum_{k\in \mathcal{M}} \bar{x}_{kijfl}$ and $\tilde{x}_{kijf}=\sum_{l=1}^{q(i, j, f)} \frac{\bar{x}_{kijfl}}{\alpha_{ijf}}$ for all $k\in \mathcal{M}$
				\ENDFOR
				\FOR{$l\in [1,L]$}
				    \STATE let $load_{I}(i,h)$ be the load on the $i$-th input port of the network core $h$
				    \STATE let $load_{O}(j,h)$ be the load on the $j$-th output port of the network core $h$
				    \STATE let $\mathcal{A}_{lh}$ be the set of coflows allocated to network core $h$ in the $l$-th interval
				    \STATE both $load_{I}$ and $load_{O}$ are initialized to zero and $\mathcal{A}_{lh}=\emptyset$ for all $h\in [1, m]$
				    \FOR{every flow $(i, j, f)$ in non-decreasing order of $\bar{C}_f$, breaking ties arbitrarily}
				        \STATE $\tilde{x}_{M_kijf}=\sum_{u\in M_k} \tilde{x}_{uijf}$ for all $k=1,\ldots, K$
				        \STATE $\ell_{ijf}=\max_{\ell\in [1,K]} \ell$ s.t. $\sum_{k=\ell}^{K} \tilde{x}_{M_kijf}\geq 1/2$
						    \STATE $r(i, j, f)=\arg_{M_k:\ell_{ijf}\leq k\leq K} \max s(M_k)$
						    \STATE $h^*=\arg \min_{h\in M_{r(i, j, f)}}\frac{1}{s_h}\left(load_{I}(i,h)+load_{O}(j,h)\right)$
						    \STATE $\mathcal{A}_{lh^*}=\mathcal{A}_{lh^*}\cup \left\{(i, j, f)\right\}$
						    \STATE $load_{I}(i,h^*)=load_{I}(i,h^*)+d_{ijf}$ and $load_{O}(j,h^*)=load_{O}(j,h^*)+d_{ijf}$
				    \ENDFOR
				\ENDFOR
		    \FOR{each $k\in \mathcal{M}$ do in parallel}
    				\FOR{$l\in [1,L]$}
				        \STATE wait until the first coflow is released
						    \WHILE{there is some incomplete flow}
						        \STATE for all $(i, j, f)\in \mathcal{A}_{lk}$, list the released and incomplete flows respecting the non-decreasing order in $\bar{C}_{ijf}$
								    \STATE let $L$ be the set of flows in the list
                    \FOR{every flow $(i, j, f)\in L$}
								    		\IF{the link $(i, j)$ is idle}
								    		    \STATE schedule flow $f$
								    		\ENDIF
								    \ENDFOR
								    \WHILE{no new flow is completed or released}
								        \STATE transmit the flows that get scheduled in line 28 at maximum rate $s_k$.
								    \ENDWHILE
						    \ENDWHILE
				    \ENDFOR						
				\ENDFOR	
   \end{algorithmic}
\label{Alg4}
\end{algorithm}

\subsection{An extension for Total Weighted Completion Time}
This section gives an $O(\log m/ \log \log m)$-approximation algorithm for minimizing total weighted completion time, which is based on combining our algorithm for minimizing makespan. The method is the same as Section~\ref{sec:Algorithm2-1}. We can formulate our problem as the following linear programming relaxation.

\begin{subequations}\label{coflow:main:wc}
\begin{align}
&    \text{min} && \sum_{f\in \mathcal{F}}w_f C_{f}                           && \tag{\ref{coflow:main:wc}}\\ 
&    \text{s.t.}. && \sum_{k\in \mathcal{M}} \sum_{l=1}^{L} x_{kijfl}=1 && \forall f\in \mathcal{F}, \forall i\in \mathcal{I}, \forall j\in \mathcal{J}   \label{coflow:wc:a}\\ 
&         && \sum_{k\in \mathcal{M}} \frac{d_{ijf}}{s_k}\sum_{l=1}^{L}x_{kijfl}\leq C_{ijf}-r_f  && \forall f\in \mathcal{F}, \forall i\in \mathcal{I}, \forall j\in \mathcal{J} \label{coflow:wc:b}\\
&         && \frac{1}{s_k}\sum_{u=1}^{l} \sum_{f\in \mathcal{F}} \sum_{j\in \mathcal{J}} d_{ijf} x_{kijfu}  \leq \tau_{l} && \forall i\in \mathcal{I}, \forall k\in \mathcal{M}, \forall l\in [1, L] \label{coflow:wc:f}\\ 
&         && \frac{1}{s_k}\sum_{u=1}^{l} \sum_{f\in \mathcal{F}} \sum_{i\in \mathcal{I}} d_{ijf} x_{kijfu}  \leq \tau_{l} && \forall j\in \mathcal{J}, \forall k\in \mathcal{M}, \forall l\in [1, L] \label{coflow:wc:g}\\ 
&         && \sum_{l=1}^{L}\tau_{l-1} \sum_{k\in \mathcal{M}} x_{kijfl} \leq C_{ijf} && \forall f\in \mathcal{F}, \forall i\in \mathcal{I}, \forall j\in \mathcal{J} \label{coflow:wc:h}\\ 
&         && C_{ijf} \leq C_{f} && \forall f\in \mathcal{F}, \forall i\in \mathcal{I}, \forall j\in \mathcal{J} \label{coflow:wc:j}\\ 
&         && x_{kijfl}, C_{ijf}, C_{f}\geq 0 && \forall f\in \mathcal{F}, \forall i\in \mathcal{I}, \forall j\in \mathcal{J}, \forall k\in \mathcal{M} \label{coflow:wc:i}
\end{align}   
\end{subequations}
In the LP (\ref{coflow:main:wc}), $x_{kijfl}$ indicates whether or not flow $(i, j, f)$ completes on network core $k$ in the $l$-th interval, $C_{ijf}$ is the completion time of flow $(i, j, f)$ and $C_{f}$ is the completion time of coflow $f$.
The constraint~(\ref{coflow:wc:a}) requires all the flows must be assigned to run at some network core.
The constraint~(\ref{coflow:wc:b}) is that the time required to transmit flow $(i, j, f)$ on link $(i, j)$ cannot exceed the time period between its release and completion time. 
The constraints~(\ref{coflow:wc:f}) and (\ref{coflow:wc:g}) represent capacity limits to time $\tau_{l}$.
The constraint~(\ref{coflow:wc:h}) is a lower bound to the completion time of flow.
The constraint~(\ref{coflow:wc:j}) ensures that the completion time of coflow $f$ is bounded by all its flows.
The constraint~(\ref{coflow:wc:i}) requires the $x$ and $C$ variables to be non-negative.

Our algorithm flow-driven-list-scheduling (described in Algorithm~\ref{Alg4}) is as follows. 
Given a set of coflow $\mathcal{F}$, an optimal solution $\bar{C}$ and $\bar{x}$ can be obtained by the linear program (\ref{coflow:main:wc}). 
Lines 1-5 schedule all flow into time intervals and normalize the value of $\bar{x}$. 
Lines 6-19 are the flow-makespan-list-scheduling algorithm, which schedules the flows in the corresponding time interval to each network core.
Lines 20-36 transmit all coflow. Same as Section~\ref{sec:Algorithm2-1}, we have the following theorem.
\begin{thm}
The flow-driven-list-scheduling has an approximation ratio of, at most, $(64K+32\gamma)=O(\log m/ \log \log m)$.
\end{thm}

\section{Results and Discussion}\label{sec:Results}
This section compares the approximation ratio of the proposed algorithm to that of the previous algorithm. 
We compares with the algorithm of Huang \textit{et al.}~\cite{Huang2020}, which schedules a single coflow on a heterogeneous parallel network.
In the scheduling single divisible coflow problem, our algorithm achieves an approximation ratio of $(4K+2\gamma)$ where $\gamma= \log m/\log \log m$ and $K= \left\lceil \log_{\gamma} m \right\rceil$.
Figure~\ref{fig:ratio1} presents the numerical results concerning the approximation ratio of algorithms.
When $m\geq 25$, the proposed algorithm outperforms the algorithm in \cite{Huang2020}.

\begin{figure}[!ht]
    \centering
        \includegraphics[width=6.4in]{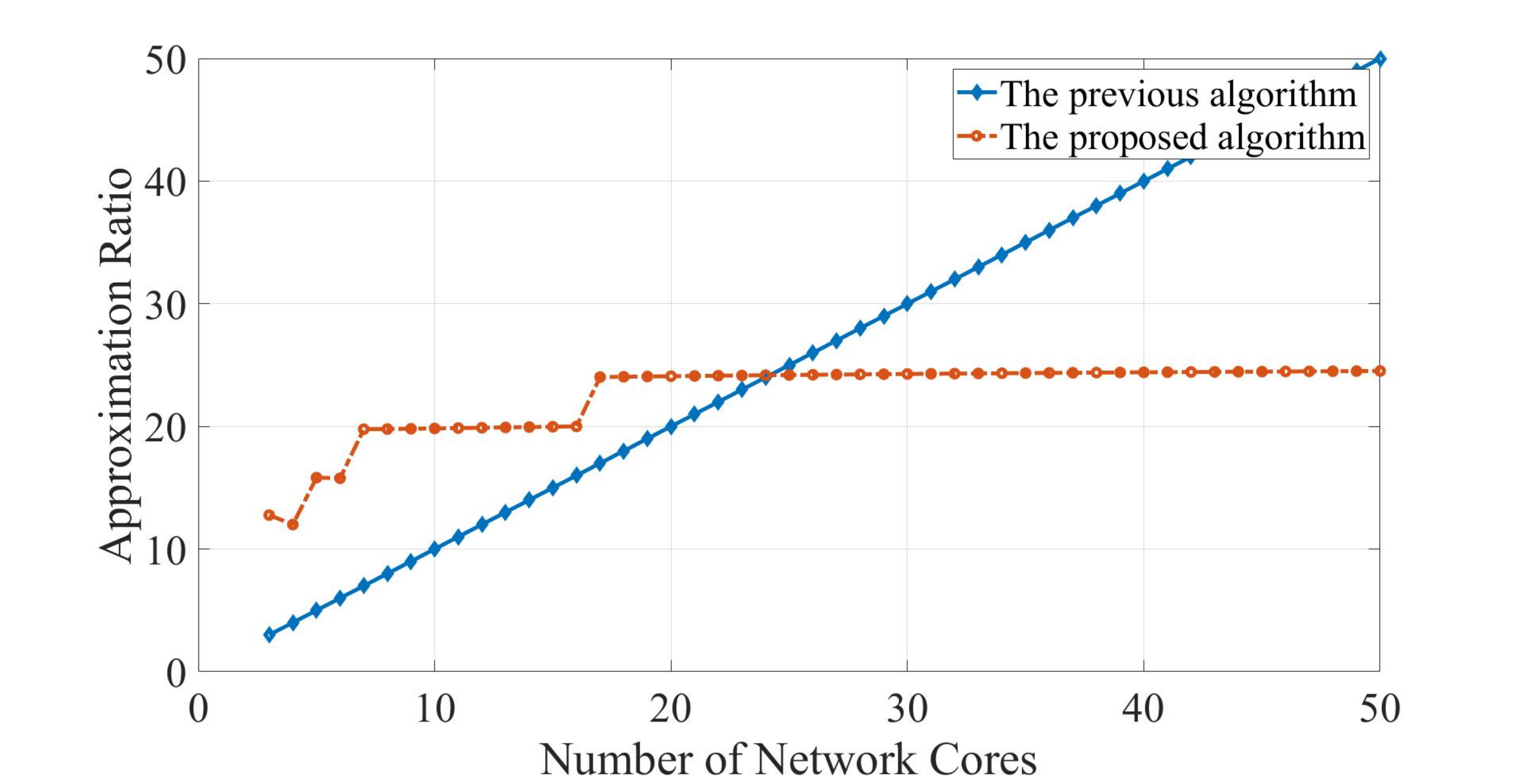}
    \caption{The approximation ratio between the algorithm in \cite{Huang2020} and the proposed algorithm.}
    \label{fig:ratio1}
\end{figure}

\section{Concluding Remarks}\label{sec:Conclusion}
With the growth of data centers, the scheduling of the single-core model is no longer sufficient. Therefore, we consider scheduling coflow problems in heterogeneous parallel networks. In this paper, two polynomial-time approximation algorithms are developed for scheduling divisible and indivisible coflows in heterogeneous parallel networks, respectively. Considering the divisible coflow scheduling problem, the proposed algorithm achieve an approximation ratio of $O(\log m/ \log \log m)$ with arbitrary release times, where $m$ is the number of network cores. On the other hand, when coflow is indivisible, the proposed algorithm achieve an approximation ratio of $O\left(m\left(\log m/ \log \log m\right)^2\right)$ with arbitrary release times.

\end{document}